\documentclass[preprint,12pt]{elsarticle}

\usepackage{amssymb}
\usepackage{amsmath}
\usepackage{multirow}

\usepackage{xcolor}
\usepackage{amsthm}
\usepackage{bm}

\def\d{\partial}
\newcommand\dd{\, \mathrm{d}}

\newcommand\xdoi[1]{#1}

\usepackage{multirow}
\usepackage{siunitx}
\usepackage{makecell}

\theoremstyle{plain}
\newtheorem{mth}{Theorem}[section]

\begin{document}

\begin{frontmatter}

\title{Second-order supporting quadric method for designing freeform refracting surfaces generating prescribed irradiance distributions} 

\author[label1,label2]{Albert~A.~Mingazov}
\author[label1,label2]{Dmitry~A.~Bykov}
\author[label1,label2]{Evgeni~A.~Bezus}
\author[label1,label2]{Leonid~L.~Doskolovich}

\affiliation[label1]{organization={Samara National Research University},
 addressline={34 Moskovskoye shosse}, 
 city={Samara},
 postcode={443086}, 
 country={Russia}}
						
\affiliation[label2]{organization={Image Processing Systems Institute, National Research Centre ``Kurchatov Institute''},
 addressline={151 Molodogvardeyskaya st.}, 
 city={Samara},
 postcode={443001}, 
 country={Russia}}

\begin{abstract}
We consider the inverse problem of calculating a refracting surface that generates a prescribed irradiance distribution in the far field for a collimated incident beam.
This problem can be formulated as a mass transportation problem (MTP) with a quadratic cost function.
To solve this problem, we propose a version of the supporting quadric method (SQM), in which the calculation of the quadric parameters is reduced to the problem of minimizing a convex function.
We obtain simple analytical expressions for the second derivatives of this function, making it possible to 
calculate the quadric parameters using second-order optimization methods.
This allows us to refer to the proposed method as the second-order SQM.
We demonstrate high efficiency of this approach by designing several optical surfaces that generate complex irradiance distributions.
We also consider the application of the second-order SQM to nonimaging optics problems described by MTPs with a non-quadratic cost function.
\end{abstract}


\end{frontmatter}

\section{Introduction}

The problem of calculating a refracting or reflecting optical surface from the condition of generating a prescribed irradiance distribution in a certain region has many important practical applications in the design of various lighting and backlight systems~\cite{44,45,46,47}.
This problem belongs to the class of inverse problems of nonimaging optics and is usually solved in the geometrical optics approximation.
One of the most universal and widely used approaches for solving this problem is the so-called supporting quadric method (SQM)~\cite{7,8,9,10,11,12}.
In the general case, the SQM includes the following steps:
(1)~approximating the irradiance distribution that has to be generated by a discrete distribution, i.e., by a finite set of points, 
(2)~representing the optical surface in a special parametric form corresponding to an upper or lower envelope of quadrics focusing the incident beam to the points of the discrete distribution, and 
(3)~applying a certain method for optimizing the quadric parameters from the condition of forming the required irradiance values at the focusing points.
Depending on the problem being solved, paraboloids of revolution~\cite{7,8}, ellipsoids and hyperboloids of revolution~\cite{25,28}, as well as more complex surfaces~\cite{27,29} can be used as quadrics.

In the problem of calculating a refracting surface that generates a prescribed far-field irradiance distribution (or a given radiant intensity distribution) for a collimated incident beam, which is considered below, the quadrics are ``degenerate'' and are represented by planes~\cite{3}.
The planes are oriented so that the incident beam is redirected upon refraction to the points of the target region.
The parameters of the quadrics, which determine the irradiance values at the points of this region, are the distances from the planes to the origin of coordinates.

The SQM is most efficient in the problems of calculating optical surfaces, which can be formulated as the Monge--Kantorovich mass transportation problem with a certain cost function~\cite{11,38}.
In this case, the calculation of the parameters of the quadrics that constitute the optical surface is reduced to finding the extremum of a concave or convex function depending on the quadric parameters~\cite{11}.
Since in Ref.~\cite{12}, the specified function was obtained from the Lagrange functional for the mass transportation problem (MTP), in what follows, we will refer to this function as the modified Lagrange function (MLF).
It is important to note that the gradient of this function is given by a simple analytical expression, and, accordingly, various gradient methods can be used to find the extremum of the MLF (i.e., to calculate the quadric parameters)~\cite{11,12}.
Moreover, due to the concavity (or convexity) of the function being optimized, it has no more than a single extremum, which can only be its global maximum (or minimum).
Since the gradient methods are first-order optimization methods, the corresponding versions of the SQM~\cite{11,12} can be referred to as the first-order supporting quadric method.

A number of important inverse problems of optical surface design can be formulated as MTPs with a quadratic cost function.
In particular, this is the case for the following problems:
calculating a refracting (or reflecting) surface that generates a prescribed irradiance distribution in the far field for a collimated incident beam~\cite{3}, 
calculating a two-mirror system for the so-called collimated beam shaping~\cite{9}, 
and calculating the eikonal function of the light field that provides focusing to a given region in the paraxial approximation~\cite{4}.
MTPs with a quadratic cost function possess a number of unique properties related to the existence and form of the optimal transport plan~\cite{A1}.
Therefore, one can assume that for solving the above-mentioned inverse problems of calculating optical surfaces, special versions of the SQM can be developed that have higher efficiency compared to the known first-order SQM.
In particular, the development of second-order SQM based on the calculation of quadric parameters using second-order optimization methods is of great interest.

In the present work, by considering as an example the problem of calculating a refracting surface, which generates a prescribed irradiance distribution in the far field for a collimated incident beam, we for the first time obtain simple analytical expressions for the second derivatives of a modified Lagrange function depending on the quadric parameters.
These expressions can be efficiently evaluated and, accordingly, enable using second-order optimization methods for the refracting surface design.
It is important to note that, although the number of variables in the MLF can reach hundreds of thousands, the Hessian of this function is a symmetric sparse matrix.
This allows one to avoid the memory issues typical for second-order methods when applied to problems with a large number of variables.
The presented examples of calculating refracting surfaces demonstrate good performance of second-order methods.
In particular, using the ``benchmark'' problem of generating a uniform irradiance distribution in a square region, it is shown that the calculation of the refracting surface using the trust region method employing the Hessian reduces the required calculation time by two orders of magnitude compared to the method not using the Hessian (i.e., using only the gradient).
Good performance of the second-order SQM is also demonstrated by more complex design examples including the problem of generating a uniform irradiance distribution in the shape of an arrow on a zero background, and the problem of generating a grayscale portrait of A.~Einstein.

Let us note that the refracting surface designed to generate the arrow image is continuous and piecewise-smooth.
The piecewise-smooth form of the surface is due to the fact that in this problem, there is no continuous integrable ray mapping from the region of the incident beam to the arrow-shaped target region, which is a non-convex region with a non-smooth boundary. 
It is important to note that such problems (generation of a uniform irradiance distribution in a disconnected or non-convex region on a zero background) cannot be solved with the widely used optical surface design methods based on numerical solution of an elliptic-type nonlinear differential equation (NDE)~\cite{44,35,36,37}.
This is due to the fact that the formulation of the design problem as an NDE assumes the smoothness of the surface being calculated.
An important advantage of the SQM considered in this work is that it is not limited by this constraint.

It is also important to note that, in addition to the problem of calculating a refracting surface represented by an MTP with a quadratic cost function, the proposed method can also be applied to problems of calculating optical surfaces described by general-form MTPs with a non-quadratic cost function~\cite{11,38}.
In this case, the problem is solved iteratively on the basis of successive approximations of the general-form MTP by MTPs with quadratic cost functions, which are solved using the second-order SQM.
In this work, this approach is illustrated using an example of an MTP with a logarithmic cost function~\cite{25},
which corresponds to the calculation of a refracting surface generating a prescribed far-field irradiance distribution for the case of a spherical incident beam.

\section{The problem of calculating a refracting surface}\label{sec:2}

\subsection{Direct and inverse problems}

Let us consider the problem of calculating a refracting surface~$R$ generating a prescribed far-field irradiance distribution in the case of an incident beam with a plane wavefront.
Let us assume that the refracting surface separates two media with the refractive indices $n_1$ and $n_2$, $n_1 > n_2$.
The surface is illuminated by a plane light beam propagating in the direction $\vec{e}_z = (0, 0, 1)$ and having a certain irradiance distribution $I(\bm{u}), \bm{u} \in G$ in the plane $z = 0$, where $\bm{u} = (u_1, u_2)$ are the Cartesian coordinates in this plane (Fig.~\ref{fig:1}).

\begin{figure}
\centering\includegraphics{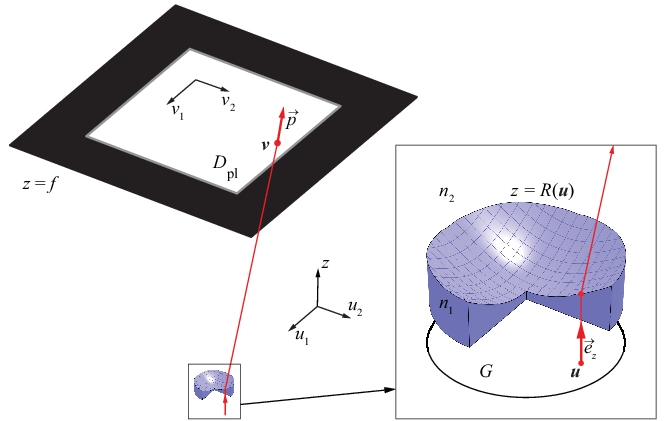}
\caption{\label{fig:1} Geometry of the problem.}
\end{figure}

Let us first consider the so-called direct problem, which consists in calculating the generated irradiance distribution in the target plane $z = f$ for a given refracting surface.
A ray originating from the point $(u_1, u_2, 0)$ in the direction $\vec{e}_z = (0, 0, 1)$ is refracted by the surface $z = R(\bm{u})$ and acquires a certain direction $\vec{p} = (p_1, p_2, p_3)$, which is determined from the Snell's law in terms of the partial derivatives of the function $R(\bm{u})$. 
This direction of the refracted ray defines a point $\bm{v} = (v_1, v_2)$ in the target plane $z = f$ as
\begin{equation}
\label{eq:3_}
\bm{v}(\vec{p}) = \left(\dfrac{p_1 f}{p_3}, \dfrac{p_2 f}{p_3}\right).
\end{equation}
Equation~\eqref{eq:3_} is written in the far field approximation.
In this approximation, we neglect the dimensions of the optical surface as compared to the distance $f$ to the target plane and assume that the rays refracted by the surface come from the point $(0, 0, 0)$, in the vicinity of which the considered surface is located.
Equation~\eqref{eq:3_} defines a ray mapping $\gamma_{\rm pl}\colon G \rightarrow D_{\rm pl}$, where $D_{\rm pl}$ is a certain region in the plane $z = f$.
In this region, an irradiance distribution $E(\bm{v})$ is generated, which is defined by the light flux conservation law
\begin{equation}
\label{cons}
\int\limits_{\gamma_{\rm pl}^{-1}(Q)} I(\bm{u})\dd\bm{u} = \int\limits_{Q} E(\bm{v})\dd\bm{v}
\end{equation}
for any Borel set $Q \subset D_{\rm pl}$.
This law can also be formulated in the following equivalent differential form~\cite{9}:
$$
I(\bm{u}) = E(\gamma_{\rm pl}(\bm{u})) \cdot J_{\gamma_{\rm pl}}(\bm{u}),
$$  
where $J_{\gamma_{\rm pl}}(\bm{u})$ is the Jacobian of the mapping $\gamma_{\rm pl}$ calculated at the point $\bm{u}$.

Let us now consider the inverse problem of calculating a refractive optical element.
Assume that we are given certain distributions $I(\bm{u})$ and $E(\bm{v})$ having the same total light flux:
\begin{equation}
\int\limits_{G} I (\bm{u})\dd\bm{u} = \int\limits_{D_{\rm pl}} E(\bm{v})\dd\bm{v}.
\end{equation}
Then, the inverse problem to be solved consists in calculating a surface $z = R(\bm{u})$, for which the corresponding ray mapping $\gamma_{\rm pl}$  satisfies Eq.~\eqref{cons}.

\subsection{Representation of the refracting surface and the mass transportation problem}

In the present subsection, we will consider the representation of the refracting surface as an envelope of a family of planes; 
such a representation is an inherent part of the supporting quadric method.
To do this, let us first write the equation of a plane, which refracts an incident light beam in the direction $\vec{p} = (v_1, v_2, f)/\sqrt{v_1^2 + v_2^2 + f^2}$ ``aimed'' at the point $(v_1, v_2, f)$ in the plane $z = f$:
\begin{equation}\label{el_surf}
z = \mathcal{K}(\bm{u}, \bm{v}) + w.
\end{equation}
Here, $w$~is an arbitrary constant and
\begin{equation}\label{cost}
\mathcal{K}(\bm{u},\bm{v}) = \dfrac{n_2 (v_1 u_1 + v_2 u_2)}{n_1 \sqrt{v_1^2 + v_2^2 + f^2} - n_2 f}.
\end{equation}
These equations can easily be obtained from the Snell's law, from which it follows that the normal to the required plane is parallel to the vector $n_1\vec{e}_z - n_2\vec{p}$. 

Let us note that Eq.~\eqref{cost} can be rewritten as a dot product.
Indeed, let us introduce new coordinates $\bm{x} = (x_1, x_2)$ in the plane $z = f$:
\begin{equation} \label{coord_x}
\left\{
\begin{array}{l}
x_1(\bm{v}) = \dfrac{n_2 v_1}{n_1 \sqrt{v_1^2 + v_2^2 + f^2} - n_2 f},\\
\\
x_2(\bm{v}) = \dfrac{n_2 v_2}{n_1 \sqrt{v_1^2 + v_2^2 + f^2} - n_2 f}.\\
\end{array}
\right.
\end{equation}
In these coordinates, Eq.~\eqref{el_surf} takes the following simple form:
\begin{equation}\label{focus_x}
z = \bm{x} \cdot \bm{u} + w.
\end{equation}
In what follows, we will use the coordinates $\bm{x}$ instead of the coordinates $\bm{v}$.
In this regard, instead of the required irradiance distribution $E(\bm{v})$, we will work with the distribution $L(\bm{x})$, which is related to $E(\bm{v})$ by the following equality:
$$
E(\bm{v}) = L(\bm{x}(\bm{v})) \cdot J_{\bm{x}}(\bm{v}),
$$
where $J_{\bm{x}}(\bm{v})$ is the Jacobian of the transformation~\eqref{coord_x} calculated at the point $\bm{v}$.
Let us remind that the required distribution $E(\bm{v})$ was defined in a certain region $D_{\rm pl}$.
The change of coordinates~\eqref{coord_x} transforms this region to a new region $D$, on which the function $L(\bm{x})$ is defined.
Taking into account this fact, we can rewrite the light flux conservation law~\eqref{cons} as
\begin{equation}
\label{cons_IL}
\int\limits_{\gamma^{-1}(Q)} I(\bm{u})\dd\bm{u} = \int\limits_{Q} L(\bm{x})\dd\bm{x},
\end{equation}
where $\gamma \colon G \to D$ is the ray mapping written in terms of the coordinates $\bm{x}$ (i.e., $\bm{x} = \gamma(\bm{u})$) and $Q$ is an arbitrary Borel subset of $D$.

When applying the supporting quadric method to the considered problem, the planes~\eqref{focus_x} that refract the incident beam in the direction of the points of the target region should be considered as quadrics.
In this case, the refracting surface $z = R(\bm{u})$ is represented as an ``upper'' envelope of the family of planes with respect to the parameters $\bm{x} \in D$
\cite{7,8,9,11,12,4,13}:
\begin{equation}\label{weak}
	R(\bm{u}) = \max\limits_{\bm{x} \in D}(\bm{x} \cdot \bm{u} + w(\bm{x})).
\end{equation}
In turn, the ray mapping reads as
 \begin{equation}\label{weak_gamma}
 \gamma(\bm{u}) = \mathop{\mathrm{argmax}}\limits_{\bm{x} \in D}(\bm{x} \cdot \bm{u} + w(\bm{x})).
 \end{equation}

Note that the maximum in Eq.~\eqref{weak} is reached at the point where the partial derivatives vanish, which enables writing the inverse mapping as
\begin{equation}\label{map_grad}
\gamma^{-1}(\bm{x}) = -\nabla_{\bm{x}} w.
\end{equation}

Let us return to the problem of calculating a refracting surface, which consists in finding a ray mapping $\gamma$ satisfying the condition of Eq.~\eqref{cons_IL}.
As it was shown in Ref.~\cite{3}, such a mapping can be found by solving the following variational problem:
\begin{equation}\label{F:scalar}
\mathcal{F}(T) = \int\limits_G \big( \bm{u} \cdot T(\bm{u}) \big) I(\bm{u})\dd\bm{u} \to \max,
\end{equation}
where the maximum has to be found over the mappings $T$ satisfying the light flux conservation law
\begin{equation}\label{cons_diff_2}
I(\bm{u}) = L(T(\bm{u})) \cdot J_T(\bm{u}),
\end{equation}
where $J_T(\bm{u})$ is the Jacobian of the mapping $T$ calculated at the point~$\bm{u}$.
Let us reiterate that the mapping denoted by $T$ is an arbitrary mapping satisfying the light flux conservation law~\eqref{cons_diff_2}, but not necessarily maximizing the functional~\eqref{F:scalar}.
Therefore, the mapping $\gamma$ introduced above is the particular mapping $T$ that maximizes this functional.

The variational problem defined by Eqs.~\eqref{F:scalar} and~\eqref{cons_diff_2} is a mass transportation problem with the cost function represented as a dot product.
Such a problem can be easily reformulated as a mass transportation problem with a quadratic cost function $\mathcal{K}(\bm{u}, \bm{x}) = \|\bm{u} - \bm{x}\|^2$ \cite{3}, however, it will be convenient for us to work with an MTP in the form of Eq.~\eqref{F:scalar}.
In the rest of the paper, when referring to a quadratic cost function, we will mean a cost function in the form of a dot product.

\section{Supporting quadric method}\label{sec:powerdiag}

When solving the problem of calculating a refracting surface using the supporting quadric method, the problem is formulated in a semi-discrete form~\cite{11,12}, in which the region $D$ is approximated by a discrete set of points $\{\bm{x}_j\}_{j = 1}^n$, whereas the required irradiance distribution $L(\bm{x})$ is replaced by a discrete set of values $L_j$, $j = \overline{1,n}$.
In this case, the envelope surface of Eq.~\eqref{weak} is replaced by the surface constituted by $n$ segments of planes~\eqref{focus_x} and thus has the following form:
\begin{equation}\label{weak_discr}
	R(\bm{u},\bm{w}) = \max\limits_{j = \overline{1,n}}(\bm{x}_j \cdot \bm{u} + w_j),
\end{equation}
where $\bm{w} = (w_j)_{j = 1}^n$ is a vector consisting of $n$ components, which we will refer to as the weights of points $\bm{x}_j$.
The representation of Eq.~\eqref{weak_discr} defines a partition of the region $G$ (the domain of definition of the incident beam) into subsets (polygons) $C^{\bm{w}}(\bm{x_j})$, which are defined by the condition
\begin{equation}\label{vcell}
	C^{\bm{w}}(\bm{x}_j) = \left\{ \bm{u} \;\big|\;
	\bm{x_j}\cdot\bm{u}+w_j\geq\bm{x}_k\cdot\bm{u}+w_k,\;\; \forall{k=\overline{1,n}}
	\right\}.
\end{equation}
At $\bm{u} \in C^{\bm{w}}(\bm{x}_j)$, the surface $z = R(\bm{u}, \bm{w})$ coincides with the plane~\eqref{focus_x}, which refracts the incident beam to the point $\bm{x}_j$, thus, the ray mapping implemented by the surface $R$ has the following form:
\begin{equation}\label{T_omega}
	T^{\bm{w}}(\bm{u}) = \bm{x}_j, \mbox{~where~} j = \mathop{\mathrm{argmax}}\limits_{k = \overline{1,n}}(\bm{x}_k \cdot \bm{u} + w_k).
\end{equation}

Since, up to a certain change of weights, the subsets $C^{\bm{w}}(\bm{x}_j)$ defined in this way are equivalent to the weighted Voronoi cells~\cite{20,merigot}, in what follows, we will refer to the regions $C^{\bm{w}}(\bm{x}_j)$ as the weighted Voronoi cells.

If the representation~\eqref{weak_discr} is adopted, then the problem of calculating a refracting surface is reduced to finding a weight vector $\bm{w}$ providing the generation of the prescribed discrete distribution:
\begin{equation}\label{cons2}
	\int\limits_{C^{\bm{w}}(\bm{x}_j)} I(\bm{u})\dd\bm{u} = L_j, ~~\forall j = \overline{1,n}.
\end{equation}

Similarly to Refs.~\cite{11,12}, one can show that the calculation of the weights $\bm{w}$ can be reduced to minimizing the following convex function:
\begin{equation}\label{h_C}
h(\bm{w}) = \sum\limits_{j = 1}^n \int\limits_{C^{\bm{w}}(\bm{x}_j)} \big(\bm{x}_j \cdot \bm{u} + w_j\big) I(\bm{u})\dd\bm{u} - \sum_{j = 1}^n w_j L_j.
\end{equation}
Let us note that the function $h$ is obtained as a modification of the Lagrange functional for the mass transportation problem~\eqref{F:scalar}, \eqref{cons_diff_2} (i.e., for an optimization problem with an equality constraint) and is exactly the modified Lagrange function mentioned in the Introduction.
Thus, for an optimal set of weights $\bm{w}$ corresponding to a minimum of the function~\eqref{h_C}, the mapping $T^{\bm{w}}$ [see Eq.~\eqref{T_omega}] approximates the solution $\gamma$ of the continuous MTP~\eqref{F:scalar}, \eqref{cons_diff_2}.

It is also important to note that, as it was discussed in the Introduction, in addition to the considered problem of calculating a refracting surface, some other important problems of nonimaging optics can also be formulated as an MTP with a quadratic cost function~\eqref{F:scalar}, \eqref{cons_diff_2}.
The solution of these problems is also reduced to finding an extremum of a function of the form~\eqref{h_C}, and, therefore, the numerical methods for solving this optimization problem considered below are quite general.

The function~\eqref{h_C} depends on a finite set of numbers $\bm{w} = (w_j)_{j = 1}^n$, thus, in order to solve the minimization problem $h(\bm{w}) \to \min$, one can use standard numerical methods.
To apply these methods efficiently, one must be able to compute not only the value of the function being minimized, but also the values of its derivatives.
As it was shown in Refs.~\cite{11,12,merigot}, these derivatives can be calculated using the following simple formula:
\begin{equation}\label{grad_ans}
\dfrac{\partial h}{\partial w_j} = S_j(\bm{w}) - L_j,
\end{equation}
where $S_j(\bm{w}) = \int\nolimits_{C^{\bm{w}}(\bm{x_j})} I(\bm{u})\dd\bm{u}$.
This formula makes it possible to use the so-called first-order optimization methods for finding the set of weights $\bm{w}$, which minimizes the function $h$.
It is easy to see that the vanishing of the derivatives at the point of minimum gives exactly the values of $\bm{w}$, which provide the generation of the prescribed discrete distribution [see Eq.~\eqref{cons2}].
As we noted in the Introduction, if a certain first-order optimization method is used for minimizing the function $h(\bm{w})$, then the method of calculating the corresponding refracting surface can be referred to as the first-order supporting quadric method.
In what follows, we will consider the use of more advanced second-order methods, which require the calculation of the Hessian being the matrix composed of the second derivatives $\partial^2 h / (\partial w_k\partial w_j)$.

\section{Calculation of the Hessian of the function $h$}

In this section, we obtain explicit expressions for the second derivatives of the function $h$ in terms of integrals over the boundaries of the weighted Voronoi cells~\eqref{vcell}.
This result is presented in the following theorem.

\begin{mth}\label{diff_cells}
	The second derivatives of the function $h$ have the form
	\begin{equation}\label{hess_ans}
		\dfrac{\d^2 h}{\d w_k\d w_j} =
		\begin{cases}
			\frac{1}{\rho_{jk}} \int\limits_{\Delta_{jk}} I(\bm{u})\dd l, &\text{if } k \neq j;\\
			- \sum\limits_{l\neq j} \frac{1}{\rho_{jl}} \int\limits_{\Delta_{jl}} I(\bm{u})\dd l, &\text{for } k = j,
		\end{cases}
	\end{equation}
	where $\Delta_{jk}$ is the common boundary between the cells $C^{\bm{w}}(\bm{x}_j)$ and $C^{\bm{w}}(\bm{x}_k)$, and $\rho_{jk} = \rho(\bm{x}_j, \bm{x}_k)$ is the distance between the points $\bm{x}_j$ and $\bm{x}_k$.
	If the cells have no common boundary, then the integral over $\Delta_{jk}$ equals zero.
\end{mth}
\begin{proof} In order to find the second derivatives of the function $h$, let us differentiate the expression~\eqref{grad_ans} with respect to $w_k$.
	The second derivative of $h$ coincides with the first derivative of the function $S_j(\bm{w})$, which is an integral over the Voronoi cell.
	Let us consider the case, when $k \neq j$, and divide the proof for this case into three parts.
	
	1. First, let us notice that when the weight $w_k$ changes, the cell $C^{\bm{w}}(\bm{x_j})$ either does not change, and then the derivative is zero (this happens when the cells $C^{\bm{w}}(\bm{x_j})$ and $C^{\bm{w}}(\bm{x_k})$ have no common boundary), or a parallel shift of one of its sides occurs.
	Let us denote the magnitude of this shift by the function $t_{jk}(w_k)$ depending on $w_k$.
	This function can be easily calculated by writing an explicit equation of the boundary.
	Let us assume that the value of the $k$-th weight changes from $w_k^0$ to $w_k$.
	Then, according to the definition of the weighted Voronoi cell~\eqref{vcell}, the side of the cell $C^{\bm{w}}(\bm{x_j})$ is shifted by an amount
	$$
	t_{jk}(w_k) = \dfrac{w_k - w_k^0}{\rho_{jk}}.
	$$
	By differentiating, we obtain that the derivative of the magnitude of the boundary shift $t_{jk}(w_k)$ with respect to the weight $w_k$ equals
	$$\dfrac{\dd t_{jk}}{\dd w_k} = \dfrac{1}{\rho_{jk}}.$$
		
	2. Let us consider a more general problem of calculating a derivative of the weighted area of a polygon with respect to a parallel shift of one of its sides.
	Assume that $W$ is a polygon on the plane, $l$ is a line containing one of its sides, $l_t$ is a line parallel to $l$ and located at a distance $t$ from it, and $W_t$ is the polygon, which is obtained from $W$ by shifting its side lying on $l$ to the line $l_t$. 
	The geometry of the shift is presented in Fig.~\ref{fig:2}.
	Let $S(t) = \int\nolimits_{W_t} I(\bm{u})\dd\bm{u}$ be the weighted area of the polygon $W_t$.
	Then, our goal consists in calculating the derivative $S'(0)$.

	\begin{figure}[h]
		\centering\includegraphics[width=0.5\linewidth]{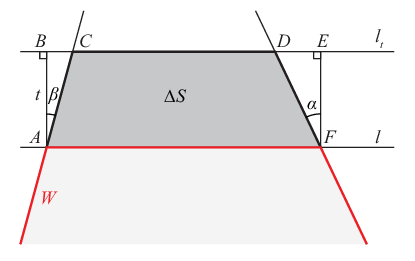}
		\caption{Polygon $W$ and increment of its area $\Delta S$ resulting from a shift of one of its sides.}
		\label{fig:2}
	\end{figure}
	
	Let us consider the increment $\Delta S = S(t) - S(0)$.
	From the geometrical point of view, it is the weighted (with the weight $I$) area of the trapezoid $ACDF$.
	It equals the difference between the area of the rectangle $ABEF$ and the areas of the triangles $ABC$ and $FED$.
	Let us consider the triangle $ABC$ and note that for the weighted area of this triangle, we can write
	$$|S_{ABC}| = \left| \int\limits_{ABC} I(\bm{u})\dd\bm{u} \right| \leq \int\limits_{ABC} |I(\bm{u})|\dd\bm{u} \leq$$
	$$\leq \max\limits_{ABC}| I(\bm{u})| \cdot \int\limits_{ABC} \dd\bm{u} = \frac{1}{2}\max\limits_{ABC} |I(\bm{u})| \cdot t^2\tan \beta = {\rm o}(t).$$
	Similarly, $S_{FED} = {\rm o}(t)$.
	For a sufficiently small shift $t$, we can consider the change in the value of the function $I$ along the direction perpendicular to the side AF to be negligible.
	Then, we can write
	$$S(t) - S(0) = t \cdot \int\limits_{AF} I(\bm{u})\dd l + {\rm o}(t),$$
	i.e.,
	$$S'(0) = \int\limits_{AF} I(\bm{u})\dd l.$$
	
	3. Let us return to the calculation of the derivative of the weighted area of the Voronoi cell in the case of $k\neq j$.
	By using the chain rule for the differentiation, we obtain:
	$$
	\dfrac{\d S_j}{\d w_k} = \dfrac{\d S_j}{\d t_{jk}} \cdot \dfrac{\dd t_{jk}}{\dd w_k} = \int\limits_{\Delta_{jk}} I(\bm{u})\dd l \cdot \dfrac{1}{\rho_{jk}},
	$$
	which proves the first part of the theorem for the case of $j\neq k$.
	
	\vspace{0.5em}
	It remains to consider the case of $j = k$ and to calculate the derivative $\d S_j / \d w_j$.
	To do this, let us note that if we add the same value $\varepsilon$ to all weights $w_1, \ldots,w_n$, then the partition into cells will not change [see Eq.~\eqref{vcell}].
	Therefore, for any $\varepsilon$, the following equality holds:
	$$S_j(w_1 + \varepsilon, \ldots, w_n + \varepsilon) = S_j(w_1, \ldots, w_n).$$
	By differentiating this expression with respect to $\varepsilon$ and evaluating the result at $\varepsilon = 0$, we have
	$$
	\sum\limits_{k=1}^n \dfrac{\d S_j}{\d w_k} = 0,
	$$
	from where we can express $\d S_j / \d w_j$ and obtain the second statement of the theorem.
\end{proof}

Note that according to Eq.~\eqref{hess_ans}, the Hessian matrix of the function $h$ is sparse, since most of the weighted Voronoi cells do not share a common boundary.
This fact is important for applying a second-order method, since one of the most significant obstacles to the use of such methods in the case of a large number of variables is the high memory requirements for storing the Hessian.

\section{Numerical examples}\label{sec:numerical}

In this section, we will present several examples demonstrating the performance of the proposed SQM of the second order.
For the calculation of the optical elements (refracting surfaces), we will search for the minimum of the function~\eqref{h_C} using the trust region method employing the Hessian.
We will use the standard MATLAB implementation of this method (the function \emph{fminunc}).
To perform optimization, it is necessary to calculate the value of the optimized function~\eqref{h_C}, its gradient~\eqref{grad_ans}, and Hessian~\eqref{hess_ans}.
In the examples considered below, the irradiance distribution of the incident beam was assumed uniform in the region~$G$ (i.e., $I(\bm{u}) = \mathrm{const}$).
In this case, the value of the optimized function~\eqref{h_C}, the gradient~\eqref{grad_ans}, and the Hessian~\eqref{hess_ans} are calculated analytically once the weighted Voronoi diagram corresponding to the weight vector $\bm{w}$ is found.
It is important to note that there exist efficient methods for calculating the weighted Voronoi diagram~\cite{21}, which enable obtaining the coordinates of the vertices of the Voronoi cells.
In particular, one of the commonly used algorithms is based on computing the convex hull for a set of points in 3D space and thus has the computational complexity of ${\rm O}(n \log n)$~\cite{21}.
After the Voronoi diagram has been found, the value of the optimized function~\eqref{h_C} for the case of a uniform incident beam ($I(\bm{u})={\rm const}$) is expressed through the areas of the polygons (weighted Voronoi cells) and the barycenters of the same polygons, which are calculated analytically~\cite{polygeom}. 
The area of each weighted Voronoi cell also gives us a component of the gradient, according to Eq.~\eqref{grad_ans}.
Finally, the Hessian, according to Eq.~\eqref{hess_ans}, is expressed through the lengths of the weighted Voronoi cell edges.

Let us note that the second-order methods are most efficient when the initial approximation lies close to the extremum.
In this regard, we will use the so-called multiscale approach \cite{merigot}, which consists in the following.
A series of successive solutions on increasingly finer grids is considered, with each previous solution serving as an initial approximation for the next step.
Let us present a more formal description of the multiscale approach used in the calculations presented below.
At each step, the continuous region $D_{\rm pl}$ is replaced by a set of points $\bm{v}_j$ being the nodes of a square mesh having the size $N\times N$ and covering $D_{\rm pl}$.
At each point, we define the required energy value $E_j$.
In the coordinates $\bm{x}$, the generated distribution is defined by a different set of points $\bm{x}_j$ [see Eq.~\eqref{coord_x}] with the same energies $L_j = E_j$.
By applying the trust region method, we obtain the weights $w_j$.
Then, proceeding to the next step, we increase the number $N$ by 10\% and obtain the initial approximation for the weights by interpolating the $(\bm{x}_j, w_j)$ data from the previous step.
Using the obtained approximation, we apply the trust region method again on the refined mesh.
This mesh refinement procedure is repeated until the calculated element generates the prescribed distribution with the required quality.
In the following subsection, we compare this multiscale approach with a conventional ``single-step'' one.

\subsection{Illuminating a square region}

In this subsection, we will consider the problem of calculating a refractive element 
(refractive index of the element $n_1 = 1.5$, refractive index of the surrounding medium $n_2 = 1$), 
which generates a uniform irradiance distribution in the form of a square $D_{\rm pl}$ with the side $1000$\,mm in the distant plane $z = f = 1000$\,mm.
We will assume that the collimated beam incident on the element has a circular cross-section with the radius $R = 1$\,mm. 
At the chosen parameters, we can neglect the dimensions of the incident beam so that the far-field approximation discussed above can be used.

Table~\ref{tab:1} shows the time required for minimizing the function~\eqref{h_C} using different numerical methods with different mesh sizes.
Let us first consider the solution of the problem without the multiscale approach.
As it is evident from the table, the trust region method, when not supplied with the Hessian, allows one to deal with meshes with sizes no larger than $20\times20$ and thus is not applicable to practically significant problems.
When, however, we calculate the Hessian using Eq.~\eqref{hess_ans} and pass it into the trust region method, the calculation becomes orders-of-magnitude faster, allowing one to reach the $100\times100$ mesh size.
The last column of Table~\ref{tab:1} shows the result of a quasi-Newton method, BFGS, where the Hessian is estimated numerically (thus, Eq.~\eqref{hess_ans} is not used).
Comparing to the trust region method (supplied with the Hessian), the quasi-Newton method is somewhat slower at small meshes, but becomes more than three times faster for a $100\times100$ mesh.

The most important part of the table is its last row, in which the time required for solving the minimization problem $h(\bm{w}) \to \min$ using the multiscale approach is presented.
While one can see an order-of-magnitude speed-up for the quasi-Newton method, it is the trust region method supplied with the Hessian that demonstrates the best results, allowing one to solve the considered problem in 8~seconds (i.e., 10 times faster than using the BFGS method).
This is the case because the multiscale approach gives a good initial approximation, which is relatively close to the extremum, so that the second-order approximation of the objective function becomes relevant.
Thus, it makes sense to use second-order methods utilizing the Hessian in conjunction with the proposed multiscale approach, which leads to an orders-of-magnitude speed-up compared to other methods.

\begin{table}[tbh]
	\caption{\label{tab:1}Time required for solving the optimization problem $h(\bm{w}) \to \min$ using different numerical methods}
	\centering
		\begin{tabular}{cS[table-format=3.2]S[table-format=3.2]S[table-format=3.2]} 
		 \hline\hline
			\multirow{3}{*}{Number of points} & \multicolumn{3}{c}{Time, s}                                 \\ \cline{2-4} 
	                                     & {Trust region} & {\makecell{Trust region\\ with Hessian}} & {\makecell{Quasi-Newton\\ (BFGS)}} \\																		
			\hline
			$10\times 10$ & 19  & 0.29 & 0.58 \\
			$15\times 15$ & 123 & 0.63 & 1.3  \\
			$20\times 20$ & 524 & 1.2  & 3.6  \\
			$50\times 50$ & {---} & 40   & 64   \\ 
			$100\times 100$& {---} & 1887 & 541  \\ 
			\hline
			$100\times 100$, multiscale & 
			        {---}  & 8.4  & 83 \\
			\hline\hline
		\end{tabular}
\end{table}

The optical element corresponding to the last row of the table and calculated using the multiscale approach with the trust region method is presented in Fig.~\ref{fig:4}(a).
The dimensions of the designed element are the following: radius $R = 1$\,mm, thickness $0.494$\,mm.
For calculating the surface of the element, Eq.~\eqref{weak} was used, where the function $w(\bm{x})$ was calculated using a spline interpolation of the weights $w_j = w(\bm{x}_j)$.
The irradiance distribution generated by the designed element is shown in Fig.~\ref{fig:4}(b).
This distribution was calculated using an in-house implementation of the ray-tracing technique with $10^7$ rays.
The luminous efficiency of the designed element (i.e., the fraction of energy of the incident beam directed by the element to the target square region) calculated taking into account the Fresnel losses amounts to 91\%.
The normalized root-mean-square deviation of the generated distribution from the required uniform distribution equals 5.1\%, which demonstrates high performance of the proposed  second-order SQM.

\begin{figure}[h]
\centering\includegraphics{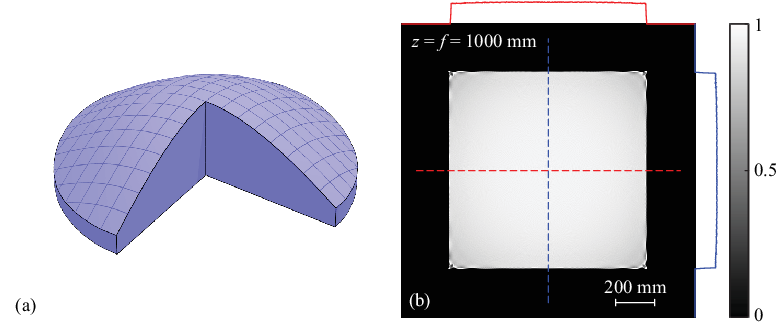}
\caption{Designed element (a) and generated normalized irradiance distribution (b).
Cross-sections along the dashed lines are also shown in (b).}
\label{fig:4}
\end{figure}

\subsection{Illuminating a non-convex region and generating a grayscale image}

In this subsection, we consider two more complex examples.
As the first example, we design an optical element illuminating an arrow-shaped region (see Fig.~\ref{fig:5}).
We used the following design parameters: distance to the target plane $f = 1000$\,mm, arrow ``span'' along the horizontal direction $1000$\,mm, arrow line width 170\,mm, radius of the incident beam $R = 1$\,mm, number of points representing the arrow at the last step of the multiscale procedure: 62000.
The calculation of the optical surface by the second-order SQM took of about 21 minutes.
The designed element is shown in Fig.~\ref{fig:5}(a) and has a radius of 1\,mm and thickness of 0.686\,mm.
The generated irradiance distribution calculated using the ray-tracing technique with $10^7$ rays is presented in Fig.~\ref{fig:5}(b).
The luminous efficiency of the calculated element equals 91\%, whereas the normalized root-mean-square deviation of the generated distribution from the required one amounts to 10.8\%. 
A smaller value of the root-mean-square deviation can be obtained by using a larger number of points $N$ in the design.
Let us note that the considered arrow-shaped region is non-convex, therefore, the ray mapping $G \to D_{\rm pl}$ is discontinuous.
As a result, the designed refractive surface, which is shown in Fig.~\ref{fig:5}(a), is not smooth but continuous and piecewise-smooth.
The presented example demonstrates that the described method remains applicable in this case, in contrast to widely used methods based on solving an elliptic-type NDE, which require introducing a non-zero background to be applied~\cite{44,35,36,37}.

\begin{figure}[h]
\centering\includegraphics{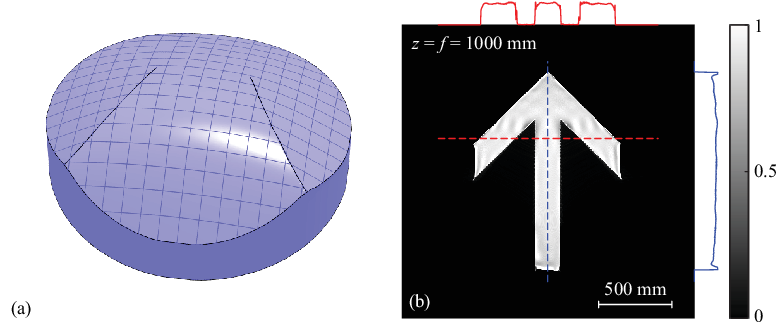}
\caption{Designed element~(a) and generated normalized arrow-shaped irradiance distribution in the plane $z = f = 1000\,\mathrm{mm}$.
Cross-sections along the dashed lines are also shown in (b).}
\label{fig:5}
\end{figure}

As the last example, we designed an element that generates a complex grayscale image, namely, a photo of A.~Einstein (see Fig.~\ref{fig:6}).
The size of the generated image in the $z = f = 1000$\,mm plane was assumed to be $500$\,mm.
Other parameters are the following: radius of the incident beam $R = 1$\,mm, number of points at the last step $560^2 = 313600$.
As before, at each step of the multiscale procedure, we used a uniform square mesh, however, the required energies $L_j$ at each point were different, since a non-uniform distribution had to be generated.
At each step, the values of $L_j$ were obtained as pixel brightness values of a properly scaled original image.
The time required for the element calculation amounted to 85~minutes.
The designed optical element (radius 1\,mm, thickness 0.248\,mm) and the generated irradiance distribution are presented in Fig.~\ref{fig:6}.
The calculated element has a luminous efficiency of 95\% and normalized root-mean-square error of 9.6\%.
This example shows that the proposed method enables generating quite complex irradiance distributions.

\begin{figure}[h]
\centering\includegraphics{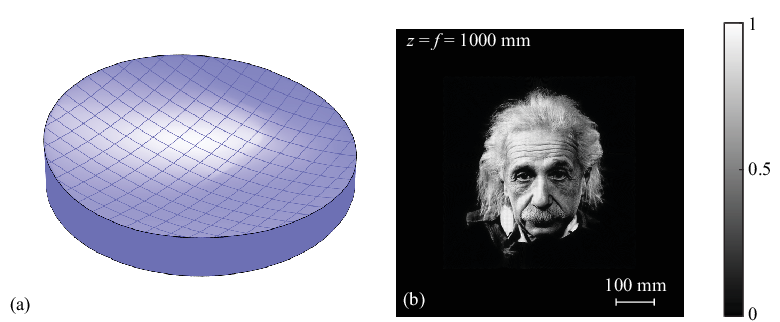}
\caption{Designed element~(a) and generated normalized grayscale irradiance distribution in the plane $z = f = 1000\,\mathrm{mm}$ (b).}
\label{fig:6}
\end{figure}

The examples considered above demonstrate high performance of the proposed second-order SQM in the problem of calculating a refracting surface.
It is important to note that with minimal modifications, the method can be applied to several other problems of nonimaging optics where the cost function is also quadratic, namely, the collimated-to-far-field reflector design problem and the design of a mirror pair for collimated beam shaping.
Several paraxial problems (e.g., the calculation of an eikonal function) can also be formulated as an MTP with a quadratic cost function and, therefore, solved using the proposed second-order approach.
Interestingly, the second-order SQM can also be applied to some problems with non-quadratic cost functions, as we demonstrate in the following section.

\section{Application to nonimaging optics problems with a non-quadratic cost function}

In the second-order supporting quadric method considered above, it was assumed that the problem of calculating an optical surface can be formulated as an MTP with a quadratic cost function.
In the present section, we will consider the case, in which the problem of calculating an optical surface is described by an MTP with an arbitrary, non-quadratic cost function $\mathcal{K}(\bm{u},\bm{x})$.
The inverse problems of optical surface design, which can be formulated as MTPs with different cost functions, are listed in Ref.~\cite{11}.
For a non-quadratic cost function, instead of cells with straight boundaries $C^{\bm{w}}(\bm{x}_j)$ [see Eq.~\eqref{vcell}], ``generalized'' Voronoi cells with curvilinear boundaries arise~\cite{12}, which does not allow for efficient and accurate calculation of the Hessian.
Nevertheless, as we will show in the present section, solving an MTP with a non-quadratic cost function can be reduced to sequential solution of several MTPs with a quadratic cost function.
In this approach, the non-quadratic cost function $\mathcal{K}(\bm{u},\bm{x})$ is approximated by a scalar product at each step.

\subsection{Quadratic approximation of the cost function}\label{sec:approax}

Let us assume that the problem of calculating an optical element can be reformulated as an MTP with a certain non-quadratic cost function $\mathcal{K}$.
In this case, it is necessary to find a mapping $P\colon D \rightarrow G$, which maximizes the functional
\begin{equation}\label{func_nn}
\mathcal{F}(P) = \int\limits_D \mathcal{K}(P(\bm{x}), \bm{x})L(\bm{x})\dd\bm{x} \to \max
\end{equation}
among mappings satisfying an analogue of the condition~\eqref{cons_IL}
\begin{equation}
\label{cons_IL2}
\int\limits_{W} I(\bm{u})\dd\bm{u} = \int\limits_{P^{-1}(W)} L(\bm{x})\dd\bm{x}
\end{equation}
for any Borel subset $W \subset G$.
Here, it is convenient for us to consider an MTP in the ``reverse direction'', i.e., from the target region $D$ containing the generated irradiance distribution to the region $G$ of the incident beam.
This is due to the fact that the change of coordinates used below and describing the transition to a quadratic MTP can be conveniently performed in the target domain.

Assume that we know a certain mapping $P^k\colon D \rightarrow G$, which is close to the ``optimal'' mapping maximizing the functional~\eqref{func_nn}.
Then, we can consider an approximation of the problem of Eqs.~\eqref{func_nn} and~\eqref{cons_IL2} in the vicinity of this mapping.
Let us introduce the quantity $\bm{\varepsilon}(\bm{x}) = P^{k+1}(\bm{x}) - P^{k}(\bm{x})$, where $P^{k+1}\colon D \rightarrow G$ is the next approximation of the sought-for mapping.
If the approximation $P^{k}$ is good enough, we can expect $\bm{\varepsilon}(\bm{x})$ to be small and replace the cost function with its Taylor series expansion in the vicinity of $\bm{\varepsilon} = 0$ up to the linear term:
$$
\begin{aligned}
\mathcal{K}\big(P^{k+1}(\bm{x}), \bm{x}\big)
&= \mathcal{K}\big(P^{k}(\bm{x}) + \bm{\varepsilon}(\bm{x}), \bm{x} \big)
\\& \approx \mathcal{K}\big(P^{k}(\bm{x}), \bm{x}\big) + \bm{\varepsilon}(\bm{x}) \cdot \nabla_{\bm{u}}\mathcal{K}\big(P^{k}(\bm{x}), \bm{x}\big)
\\& = \mathcal{K}\big(P^{k}(\bm{x}), \bm{x}\big)
- P^{k}(\bm{x})\cdot \nabla_{\bm{u}}\mathcal{K}\big(P^{k}(\bm{x}), \bm{x}\big)
\\& + P^{k+1}(\bm{x})\cdot \nabla_{\bm{u}}\mathcal{K}\big(P^{k}(\bm{x}), \bm{x}\big),
\end{aligned}
$$
where ``$\cdot$'' denotes the dot product of 2-vectors and $\nabla_{\bm{u}} \mathcal{K}$ denotes the gradient of $\mathcal{K}(\bm{u},\bm{x})$ calculated with respect to its first argument.

Let us substitute the obtained approximation of the cost function into the functional~\eqref{func_nn}:
$$
\begin{aligned}
\mathcal{F}(P^{k+1}) \approx
&\int\limits_D \mathcal{K}\big(P^{k}(\bm{x}), \bm{x}\big)L(\bm{x})\dd\bm{x}
- \int\limits_D P^{k}(\bm{x}) \cdot \nabla_{\bm{u}}\mathcal{K}\big(P^{k}(\bm{x}), \bm{x}\big)L(\bm{x})\dd\bm{x}
\\
& + \int\limits_D P^{k+1}(\bm{x}) \cdot \nabla_{\bm{u}}\mathcal{K}\big(P^{k}(\bm{x}), \bm{x}\big)L(\bm{x})\dd\bm{x}.
\end{aligned}
$$
According to Eq.~\eqref{func_nn}, to find the next approximation $P^{k+1}$ of the mapping, we have to maximize the above functional with respect to $P^{k+1}$.
The first two terms, obviously, do not depend on $P^{k+1}$, therefore, in order to maximize $\mathcal{F}(P^{k+1})$, it is sufficient to find the maximum of the functional
\begin{equation}\label{approx_MTP}
\mathcal{G}(P^{k+1}) = \int\limits_D P^{k+1}(\bm{x}) \cdot \nabla_{\bm{u}}\mathcal{K}\big(P^{k}(\bm{x}), \bm{x}\big)L(\bm{x})\dd\bm{x} \to \max.
\end{equation}

Let us denote
\begin{equation}\label{trans_k}
 \eta(\bm{x}) = \nabla_{\bm{u}}\mathcal{K}\big(P^{k}(\bm{x}), \bm{x}\big)
\end{equation}
and introduce new coordinates $\hat{\bm{x}} = \eta(\bm{x})$.
In these coordinates, the problem of Eq.~\eqref{approx_MTP} becomes an MTP with a quadratic cost function [see Eq.~\eqref{F:scalar}]
\begin{equation}\label{approx_dot_MTP}
\hat{\mathcal{G}}({\hat P}^{k+1}) = \int\limits_{\hat D} {\hat P}^{k+1}(\hat{\bm{x}}) \cdot \hat{\bm{x}} \, \hat L(\hat{\bm{x}})\dd\hat{\bm{x}} \to \max.
\end{equation}
This MTP is formulated not for $L(\bm{x})$ but for a different target distribution $\widehat{L}(\hat{\bm{x}}) = L(\bm{x}) \cdot J^{-1}$, where $J$ is the Jacobian of the transformation~\eqref{trans_k}.
Let us also note that the mapping ${\hat P}^{k+1}$ is the mapping $P^{k+1}$ in the coordinates $\hat{\bm{x}}$, i.e., $P^{k+1}(\bm{x}) = \hat{P}^{k+1}(\eta(\bm{x}))$, and the region $\hat{D}$ is the image of the region $D$ under the mapping $\eta$.

For solving the MTP of Eq.~\eqref{approx_dot_MTP} with a quadratic cost function, we can apply the second-order method of Eqs.~\eqref{h_C}, \eqref{grad_ans}, and~\eqref{hess_ans} described above.

If the iterative process converges, then, at some point, the obtained mapping will stop changing, i.e., $P^{k+1} = P^{k}\colon D \rightarrow G$.
It is easy to see that this means that $\hat{P}^k$ is a solution of the mass transportation problem with a quadratic cost function~\eqref{approx_dot_MTP} for the mappings from $\hat{D}$ to $G$.
According to the theorem proven in~\ref{subsec:appa1}, in this case, $P^{k}$ is the solution of the initial mass transportation problem of Eqs.~\eqref{func_nn} and~\eqref{cons_IL2} with a non-quadratic cost function.
This justifies the utilization of the proposed iterative method for solving MTPs with arbitrary cost functions, which includes its successive approximations by MTPs with quadratic cost functions.
A more detailed description of the iterative algorithm is given in~\ref{subsec:appa2}.

\subsection{Numerical example}

Let us consider the application of the described iterative approach to the problem of designing a refractive optical element, which generates a prescribed irradiance distribution in a distant plane for a point light source (i.e., in the case of a spherical incident beam).
The geometry of the problem is shown in Fig.~\ref{fig:7}.
As it was shown in Ref.~\cite{25}, this problem can be reformulated as a mass transportation problem of Eqs.~\eqref{func_nn} and~\eqref{cons_IL2} with a logarithmic cost function
$$
\mathcal{K} (\bm{u}, \bm{x}) = 
- \log \left[
1 - \frac{n_2}{n_1} \cdot 
\frac{x_1 u_1 + x_2 u_2 + f \sqrt{1 - u_1^2 - u_2^2} }{\sqrt{x_1^2 + x_2^2 + f^2}}
\right].
$$
Here, $\bm{u} = (u_1, u_2)$ is the projection onto the plane $z = 0$ of the unit vector $\vec{e}$ of the ray emitted from the source and $\bm{x} = (x_1, x_2)$ are the coordinates in the target plane $z = f$.
Under the masses in this problem, the normalized intensity distribution of the source $I(\bm{u}) = \tilde{I}(\bm{u}) / \sqrt{1 - u_1^2 - u_2^2}$, where $\tilde{I}(\bm{u})$ is the source intensity, and the prescribed intensity distribution $L(\bm{x})$ in the target plane are understood.
Below, we will assume that the point source is Lambertian, which makes the function $I(\bm{u})$ constant.

\begin{figure}[h]
\centering\includegraphics{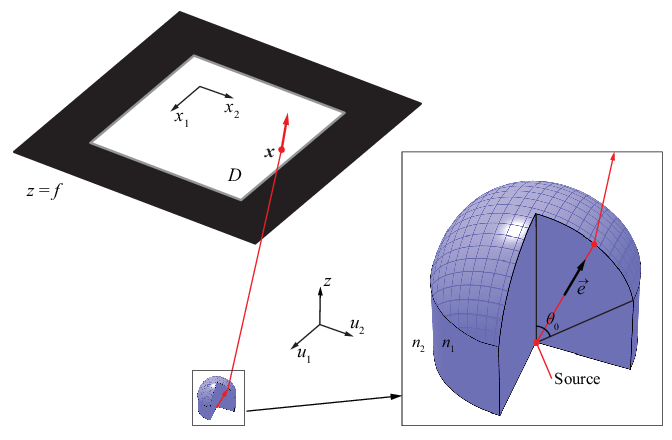}
\caption{Geometry of a refractive optical element operating with a point light source.}
\label{fig:7}
\end{figure}

Let us consider the design of a refractive optical element made of a material with refractive index $n_1 = 1.5$ (refractive index of the surrounding medium $n_2 = 1$) and intended for generating a uniformly illuminated square with the side $750\,\mathrm{mm}$ in the distant plane $z = f = 1000\,\mathrm{mm}$.
We assume that the Lambertian light source radiates into a solid angle described by a cone centered at the origin and having the apex angle (the radiation angle) $2\theta_0 = \pi/2$ (see Fig.~\ref{fig:7}).

In the calculations, we used $N^2 = 100^2$ points to approximate the target square region.
The calculation of the optical element was performed using the proposed iterative method (see the description in Section~2 of the supplementary materials).
Solving the MTP at the first iteration took 9 seconds and the next iterations took of about 100 seconds each.
The results obtained at each iteration are shown in Fig.~\ref{fig:8}.
At the first iteration, the points in the target plane are placed in the initial square region~$D$.
At the subsequent iterations, the points are shifted according to the discrete analogue of Eq.~\eqref{trans_k} [see Eq.~(A.3) in~\ref{subsec:appa2}].
It is evident from Fig.~\ref{fig:8} that for the considered example, the method converges in two iterations, and already after the solution of the second quadratic MTP, the generated irradiance distribution is visually indistinguishable from the required one.
The designed optical element is shown in Fig.~\ref{fig:7}.
The overall dimensions of its ``working'' refractive surface are the following: $1.07\times 1.07 \times 0.458\,\mathrm{mm}^3$, while the distance from the light source to the upper point of the designed element was fixed at $1\,\mathrm{mm}$.
The energy efficiency of the designed element (calculated taking into account the Fresnel losses) is 95\%.
The normalized root-mean-square deviation from the required uniform square distribution amounts to 6.1\%, which confirms the high accuracy of solving the considered problem using the proposed approach.

\begin{figure}[h]
\centering\includegraphics{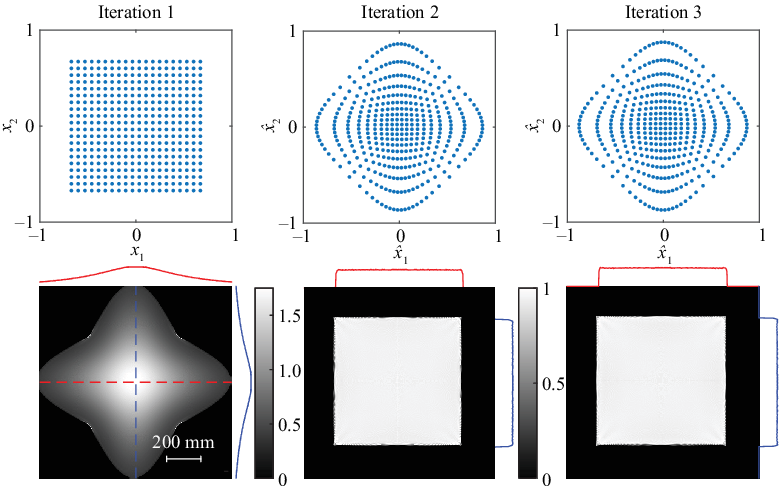}
\caption{Points in the target region (upper plots) and generated normalized irradiance distributions in the plane $z = f = 1000\,\mathrm{mm}$ (lower plots) at different iterations.}
\label{fig:8}
\end{figure}

\section{Conclusion}

In the present work, we considered the inverse problem of calculating a refracting surface that generates a prescribed irradiance distribution in the far field for a collimated incident beam.
To solve this problem, we proposed the second-order supporting quadric method.
In the method, the surface is represented as a set of planes (degenerate quadrics) that refract the incident beam toward the points of the region being illuminated.
The calculation of the quadric parameters (distances from the planes to the origin of coordinates) was reduced to the problem of minimizing a convex function,
which, along with its gradient, can be computed by integrating the irradiance distribution of the incident beam over the cells of a weighted Voronoi diagram.
We obtained explicit analytical expressions for the Hessian of the optimized function.
It is important to note that the Hessian matrix, the elements of which are calculated as integrals over common boundaries of the Voronoi cell pairs,
has a symmetric sparse form, which allows for efficient calculation of the quadric parameters using second-order optimization methods.
Excellent performance of the second-order supporting quadric method was illustrated by several design examples.
In particular, by considering a ``benchmark'' problem of generating a uniform irradiance distribution in a square region, we demonstrated that the calculation of the refracting surface using the trust region method employing the Hessian enables reducing the computation time by two orders of magnitude compared to the first-order method.
A multiscale version of the SQM was also considered, and using the example of the mentioned benchmark problem, it was shown that this method is by an order of magnitude faster than a similar multiscale approach based on the quasi-Newton BFGS method.

It is important to note that the proposed second-order SQM is quite general and can be applied to other optical design problems that can be formulated as a mass transportation problem with a quadratic cost function.
This is due to the fact that the calculation of the quadric parameters in such problems also reduces to finding the extremum of a convex function, which has the same form as in the problem of calculating a refracting surface discussed above.
Moreover, the proposed method can also be applied to optical design problems described by general-form MTPs with a non-quadratic cost function.
In this case, the problem can be solved iteratively on the basis of successive approximations of general-form MTPs with MTPs having quadratic cost functions, which are solved using the second-order SQM. 
High efficiency of this approach was confirmed by an example of calculating a refracting surface generating a uniform irradiance distribution in the far field for a spherical incident beam.
Such a problem is described by an MTP with a logarithmic cost function, and in the presented example, the proposed iterative method converged in only two iterations.

\section*{Acknowledgments}
This work was funded by the Russian Science Foundation (project no. 24-19-00080).

\appendix
\section{Supplementary materials}
\label{sec:app}

\subsection{Approximation of an MTP with an arbitrary cost function by an MTP with a quadratic cost function}\label{subsec:appa1}

\begin{mth}\label{extrem_approx}
	Let the mapping $P\colon D \rightarrow G$ be such that the mapping $\eta\colon D \rightarrow \hat{D}$ defined by $\eta(\bm{x}) = \nabla_{\bm{u}}\mathcal{K}\big(P(\bm{x}), \bm{x}\big)$ is invertible.
	If the mapping $\hat{P}\colon \hat{D} \rightarrow G$ defined by the condition $P = \hat{P}\circ \eta$ is a solution of an MTP with a quadratic cost function of Eq.~(25), then the mapping $P$ is an extremum of the MTP of Eqs.~(21) and~(22) with the cost function $\mathcal{K}(\bm{u},\bm{x})$.
\end{mth}
\begin{proof}
	Obviously, the light flux conservation law for $P$ is equivalent to the light flux conservation law for $\hat{P}$.
	Then, it is sufficient to check the equivalence of the so-called integrability conditions [9,10], which are obtained as the extremum condition of the Lagrange functional of the corresponding MTP with respect to the mapping.  
	The condition that the mapping $P$ is extremal for the MTP of Eqs.~(21) and~(22) can be written as
	\begin{equation}\label{diff_int}
		\dfrac{\d}{\d u_i}\Big[\mathcal{K}(\bm{u},\bm{x}) - s(\bm{u})\Big]\bigg|_{\bm{u} = P(\bm{x})} = 0,~~~~~i = \overline{1,2}
	\end{equation}
for a certain function $s(\bm{u})$, which depends on a particular form of the quadrics in the considered design problem described by the MTP [9,10]. 
	Let us perform some equivalent transformations.
	The condition~\eqref{diff_int} can be written as
	$$\nabla_{\bm{u}}\mathcal{K}(P(\bm{x}), \bm{x})  -\nabla s(P(\bm{x})) = 0$$
	or
	$$\eta(\bm{x}) - \nabla s(P(\bm{x})) = 0.$$
	Let us substitute $\bm{x} = \eta^{-1}(\hat{\bm{x}})$ to the last equality:
	$$\hat{\bm{x}} - \nabla s(\hat{P}(\hat{\bm{x}})) = 0,$$
	which, in turn, can be rewritten as
	$$\dfrac{\d}{\d u_i}\Big[\bm{u} \cdot \hat{\bm{x}} - s(\bm{u})\Big]\bigg|_{\bm{u} = \hat{P}(\hat{\bm{x}})} = 0,~~~~~i = \overline{1,2}.$$
	The last equality is obviously fulfilled if the mapping $P$ is the solution of an MTP with a quadratic cost function.
	Therefore, in this case, the corresponding mapping $P(\bm{x})$ provides the fulfillment of the extremum condition~\eqref{diff_int} for the general-form MTP.
\end{proof}

\subsection{Iterative algorithm}\label{subsec:appa2}

Let us describe the algorithm of solving a mass transportation problem with a non-quadratic cost function $\mathcal{K}$.
Let the incident beam be defined by a function $I(\bm{u})$, which is assumed to be constant in the region $G$.
The irradiance distribution $L(\bm{x})$ generated in the region $D$ can be arbitrary.
In the proposed approach, this distribution is approximated by a discrete distribution constituted by a set of points with the coordinates $\bm{x}_i$ and energies $L_i$.

\textbf{Iteration 1.}
We solve an MTP with a quadratic cost function of Eqs.~(12) and~(13) by minimizing the function of Eq.~(18) using the developed second-order method.
As a result, we obtain a set of weights $\bm{w} = (w_i)$, which approximates the values of the function $w$ at a discrete set of points.
Then, using spline interpolation over the values $w(\bm{x}_i) = w_i$, we reconstruct the continuous function $w(\bm{x})$.
The function $w$ calculated in this way allows one to find both the ``direct'' ray mapping $G \rightarrow D$ using Eq.~(10) from the main text of the article and the inverse mapping using the following formula (note that this formula is Eq.~(11) from the main text of the article written in different notation):
\begin{equation}\label{m_final}
	P^1(\bm{x}) = -\nabla w(\bm{x}).
\end{equation}

\textbf{Iteration $n$.}
We calculate the new positions of the points using Eq.~(24) from the main text of the article:
\begin{equation}\label{x_transf}
	\bm{\hat{x}_i} = \nabla_{\bm{u}}\mathcal{K}(P^{n-1}(\bm{x}_i), \bm{x}_i).
\end{equation}
Note that these points have the same energy values $L_i$ as the points $\bm{x}_i$, which define the generated distribution at the first iteration.
Next, we solve an MTP with a quadratic cost function of Eqs.~(12) and~(13) using the second-order SQM, find the optimal set of weigths $\bm{w}$, and interpolate the function $w(\hat{\bm{x}})$ using the values $w(\hat{\bm{x}}_i) = w_i$.
Then, the inverse mapping at the $n$-th iteration $P^n(\bm{x})$ can be found by interpolating its values at the following points:
$$
P^n(\bm{x}_i) = -\nabla w (\hat{\bm{x}}_i).
$$ 

If at a certain iteration, the mappings $P^M$ and $P^{M+1}$ coincide, then, according to Theorem~\ref{extrem_approx}, the solution is found.
In practical calculations, due to the discrete character of the problem, we compare the positions of the sets of points $P^{M+1}(\bm{x}_i)$ and $P^{M}(\bm{x}_i)$.

After a smooth function $P^M(\bm{x})$ is found, which defines the inverse mapping, the surface of the optical element is reconstructed using Eq.~(9) or its analogue depending on the particular nonimaging optics problem being solved.


\end{document}